\documentclass[10pt,twocolumn,twoside]{IEEEtran}
\IEEEoverridecommandlockouts

\usepackage{cite}
\usepackage{amsmath,amssymb,amsfonts}
\usepackage{algorithmic}
\usepackage{graphicx}
\usepackage{textcomp}
\usepackage{balance}
\usepackage{dsfont}
\usepackage{xcolor}
\usepackage{microtype}
\usepackage{graphicx}
\usepackage[utf8]{inputenc}

\let\HH\H

\usepackage{amsmath,amssymb,dsfont,amsthm}
\usepackage{stmaryrd}

\newtheorem{theorem}{Theorem}
\newtheorem{proposition}{Proposition}
\newtheorem{lemma}{Lemma}

\theoremstyle{definition}
\newtheorem{definition}{Definition}[section]
\newtheorem{assumption}{Assumption}

\theoremstyle{remark}
\newtheorem*{remark}{Remark}

\DeclareSymbolFont{sansops}{OT1}{\sfdefault}{m}{n}
\SetSymbolFont{sansops}{bold}{OT1}{\sfdefault}{b}{n}

\makeatletter
\renewcommand\operator@font{\mathgroup\symsansops}
\makeatother

\DeclareSymbolFont{sfoperators}{OT1}{cmss}{m}{n}
\DeclareSymbolFontAlphabet{\mathsf}{sfoperators}

\makeatletter
\def\operator@font{\mathgroup\symsfoperators}
\makeatother

\newcommand{\R}{\mathbb{R}}
\newcommand{\C}{\mathbb{C}}

\renewcommand{\Re}{\operatorname{Re}}
\renewcommand{\Im}{\operatorname{Im}}
\newcommand{\conj}[1]{\mkern 1.5mu\underline{\mkern-1.5mu#1\mkern-1.5mu}\mkern 1.5mu}

\newcommand{\E}{\operatorname{ E}}
\newcommand{\var}{\operatorname{ var}}

\renewcommand{\j}{\mathrm{j}}

\newcommand{\vct}[1]{\boldsymbol{#1}}
\newcommand{\mtx}[1]{\boldsymbol{#1}}
\newcommand{\diag}{\operatorname{diag}}

\renewcommand{\H}{\mathrm{H}}
\newcommand{\T}{\top}

\newcommand{\trace}{\operatorname{ tr}}
\newcommand{\rank}{\operatorname{ rank}}

\newcommand{\norm}[1]{\left|\left|#1\right|\right|}
\newcommand{\opnorm}[1]{\norm{#1}_{ \operatorname{op}}}
\newcommand{\Expec}[2][]{\E_{#1}\left[#2\right]}
\newcommand{\Var}[2][]{\operatorname{Var}_{#1}\left[#2\right]}

\newcommand{\p}[1]{\left(#1\right)}
\newcommand{\s}[1]{\left[#1\right]}

\newcommand{\abs}[1]{\left|#1\right|}

\newcommand{\set}[1]{\mathcal{#1}}

\newcommand{\linop}[1]{\mathcal{#1}}    

\newcommand{\va}{\vct{a}}
\newcommand{\vb}{\vct{b}}
\newcommand{\vc}{\vct{c}}

\newcommand{\ve}{\vct{e}}
\newcommand{\vf}{\vct{f}}
\newcommand{\vg}{\vct{g}}
\newcommand{\vh}{\vct{h}}

\newcommand{\vk}{\vct{k}}

\newcommand{\vp}{\vct{p}}
\newcommand{\vq}{\vct{q}}
\newcommand{\vr}{\vct{r}}
\newcommand{\vs}{\vct{s}}

\newcommand{\vu}{\vct{u}}
\newcommand{\vv}{\vct{v}}
\newcommand{\vw}{\vct{w}}
\newcommand{\vx}{\vct{x}}
\newcommand{\vy}{\vct{y}}
\newcommand{\vz}{\vct{z}}

\newcommand{\vbeta}{\vct{\beta}}
\newcommand{\vgamma}{\vct{\gamma}}

\newcommand{\vepsilon}{\vct{\epsilon}}

\newcommand{\vtheta}{\vct{\theta}}

\newcommand{\vomega}{\vct{\omega}}
\newcommand{\vzero}{\vct{0}}
\newcommand{\vone}{\vct{1}}

\newcommand{\mA}{\mtx{A}}
\newcommand{\mB}{\mtx{B}}
\newcommand{\mC}{\mtx{C}}
\newcommand{\mD}{\mtx{D}}
\newcommand{\mE}{\mtx{E}}
\newcommand{\mF}{\mtx{F}}
\newcommand{\mG}{\mtx{G}}

\newcommand{\mL}{\mtx{L}}
\newcommand{\mM}{\mtx{M}}

\newcommand{\mQ}{\mtx{Q}}
\newcommand{\mR}{\mtx{R}}
\newcommand{\mS}{\mtx{S}}

\newcommand{\mV}{\mtx{V}}
\newcommand{\mW}{\mtx{W}}
\newcommand{\mX}{\mtx{X}}
\newcommand{\mY}{\mtx{Y}}
\newcommand{\mZ}{\mtx{Z}}

\newcommand{\mUpsilon}{\mtx{\Upsilon}}

\newcommand{\mId}{\mathbf{I}}

\newcommand{\loF}{\linop{F}}

\newcommand{\setD}{\set{D}}
\newcommand{\setE}{\set{E}}

\newcommand{\setG}{\set{G}}

\newcommand{\setM}{\set{M}}
\newcommand{\setN}{\set{N}}

\newcommand{\cb}[1]{\left\{#1\right\}}

\newcommand{\intdim}{\operatorname{\sf intdim}}

\newcommand{\dUnif}{\textsc{Uniform}}
\newcommand{\dNormal}{\textsc{Normal}}
\newcommand{\pr}{\mathfrak{p}}

\newcommand{\mMtilde}{\mtx{\tilde{M}}}
\newcommand{\mYtilde}{\mtx{\tilde{Y}}}
\usepackage{hyperref}

\def\BibTeX{{\rm B\kern-.05em{\sc i\kern-.025em b}\kern-.08em
    T\kern-.1667em\lower.7ex\hbox{E}\kern-.125emX}}
\usepackage{lipsum}

\title{\LARGE \bf 
Admittance Matrix Concentration Inequalities for\\ Understanding Uncertain Power Networks 
}

\author{Samuel Talkington$^\dagger$, Cameron Khanpour$^\dagger$, Rahul K. Gupta$^\ddagger$, Sergio A. Dorado-Rojas$^\dagger$, Daniel Turizo$^\dagger$,\\ Hyeongon Park$^{\dagger\star}$, Dmitrii M. Ostrovskii$^\sharp$, Daniel K. Molzahn$^\dagger$
\vspace{-0.5cm}
 \thanks{ 
 $^\dagger$School of Electrical and Computer Engineering, Georgia Institute of Technology, Atlanta, GA, USA. Email: \{{\tt \href{mailto:talkington@gatech.edu}{talkington}, \href{mailto:ckhanpour3@gatech.edu}{ckhanpour3}, \href{mailto:sadr@gatech.edu}{sadr}, \href{mailto:djturizo@gatech.edu}{djturizo}, \href{mailto:molzahn@gatech.edu}{molzahn}}\}{\tt @gatech.edu}
}%
\thanks{$^\ddagger$School of Electrical Engineering and Computer Science, Washington State University, Pullman, WA, USA. Email: {\tt \href{mailto:rahul.k.gupta@wsu.edu}{rahul.k.gupta}@wsu.edu}}%
\thanks{$^\star$School of Systems Management and Safety Engineering, Pukyong National University,
Busan, South Korea. Email:{\tt \href{mailto:hyeongon@pknu.ac.kr}{hyeongon}@pknu.ac.kr}}%
\thanks{$^\sharp$Schools of Mathematics and Industrial and Systems Engineering, Georgia Institute of Technology, Atlanta, GA, USA. Email: {\tt \href{mailto:ostrov@gatech.edu}{ostrov}@gatech.edu} }
}

\begin{document}
\begingroup
\allowdisplaybreaks

\maketitle

\begin{abstract}
This paper presents conservative probabilistic bounds for the spectrum of the admittance matrix and classical linear power flow models under uncertain network parameters; for example, probabilistic line contingencies. Our proposed approach imports tools from probability theory, such as concentration inequalities for random matrices. This provides a theoretical framework for understanding error bounds of common approximations of the AC power flow equations under parameter uncertainty, including the DC and LinDistFlow approximations. Additionally, we show that the upper bounds scale as functions of nodal criticality. This network-theoretic quantity captures how uncertainty concentrates at critical nodes for use in contingency analysis. We validate these bounds on IEEE test networks, demonstrating that they correctly capture the scaling behavior of spectral perturbations up to conservative constants.
\end{abstract}



\section{Introduction}
\label{sec:intro}
In \textit{network problems}, the underlying graphical structure of the electric power system itself may be uncertain, estimated, controlled, or optimized. Such problems are common and challenging in many power engineering settings. In many such cases, a model is partially or completely unknown, resulting in the topology or model parameters being intrinsically uncertain. This can be a source of uncertainty in power system computations, which may cause downstream impacts on decision-making tools. 

At the same time, even in the case where the model is known with certainty, many network control problems\textemdash such as transmission switching or expansion planning\textemdash may have vast combinatorial solution spaces, which may be challenging to search through. In both settings, it is desirable to understand how such randomness \textit{propagates through the power flow equations} via the admittance matrix. To this end, we propose \textit{admittance matrix concentration inequalities} as a fundamental tool for working with random power networks models. 

\paragraph{Proposed Approach}

We consider a power network modeled by an undirected graph with~$n$ nodes and~$m$ possible (but not necessarily connected) lines. We reserve the index~$l$ for lines (edges), so that~$l \in [m] := \{1,2,\dots, m\}$, and~$l = (i,j)$ where~$i,j \in [n]$ are the indices reserved for the nodes. We denote by~$\mA \in \cb{-1,0,1}^{m \times n}$ the branch-to-bus incidence matrix, whose rows~$\cb{\va_{l}}_{l \in\s{m}}^\top$ are the \textit{incidence vectors} associated with each line~$l=(i,j)\in[m]$, where~$\va_{l} := \ve_i - \ve_j$, with~$\ve_i$ denoting the~$i$-th standard Euclidean basis vector in~$\R^n$. 
We consider a vector $\vw \in \C^m$ of \textit{random line admittances}~$w_{l}\sim \setD_{l}$, where~$\setD_{l}$ is the \textit{admittance distribution} for each line~$l=(i,j) \in [m]$, which is not necessarily assumed to be known.

This work studies \textit{random admittance matrices} of the form~$\mY := \mA^\T \mW \mA$, where~$\mW := \diag(\vw)$ is the diagonal matrix with the entries of the complex weight vector~$w \in \C^m$ on the diagonal. There are several situations where this model is useful, as discussed earlier. How do such matrices $\mY$ behave? We provide precise answers to this question under an array of assumptions; those answers come in the form of 
upper bounds for the quantities
\begin{equation*}
    \Expec[]{\|\mY\|} \quad \mathsf{and} \quad \Pr\p{\|\mY\| \geq t} \;\; \text{for some~$t > 0$,}
\end{equation*}
i.e., for the expectation and the tail probability of the operator norm~$\norm{\mY} = \sqrt{\lambda_{\sf max}\!\p{\mY^* \mY}}$ of the random admittance matrix. The admittance matrix is equivalently characterized as the \emph{graph Laplacian} of an electric grid. We direct the reader to explore recent works on the concentration of general random graph Laplacians such as \cite{oliveira_concentration_2010} and \cite{vershynin2016ConcentrationandRegularizationofRandomGraphs}. Modeling the line parameters as arbitrary bounded random variables provides a single framework that subsumes many physically relevant uncertainties (unknown parameters, probabilistic line outages, and configuration changes) directly at the level of the admittance matrix. It captures model uncertainty that other realizations of randomness considered, such as uncertainty in power injections, cannot obtain.

\paragraph{Contributions}
This paper presents conservative probabilistic upper bounds on spectral perturbations in admittance matrices. Our results make use of progress in applied probability theory; in particular, matrix concentration inequalities due to \cite{tropp_concentration_2015,bandeira_sharp_2016,HandelFreeProbability_2023}. The results address several linear power flow models used in the literature for transmission and distribution networks alike, via the linear AC power flow (LACPF) approximation due to \cite{bolognani_implicit_linear_2015,deka_structure_learning_2018}. Our main contributions are:
\begin{enumerate}
    \item \textit{Random admittance matrix framework:} We introduce the perspective of considering the edge weights of an electric grid's lines as \emph{arbitrary bounded random variables} to develop a general theory of the probabilistic behavior of power networks. 
    \item \textit{Nodal criticality and scaling characterization:} We introduce the \emph{contingency factors} $c_l$ and \emph{nodal criticality} $\Delta_{\vc}$ as network-theoretic quantities that govern spectral concentration under uncertain contingencies. We prove that the centered admittance matrix concentrates as $O(\sqrt{\Delta_{\vc} \log n})$, identifying how network topology and uncertainty jointly affect spectral behavior.
    \item \textit{Conservative bounds for risk analysis:} The bounds provide reliable upper bounds suitable for worst-case and risk-aware analysis, with the sharpest bound achieving ${\sim}1.5\times$ conservatism on IEEE networks. This framework allows us to model the DC and LinDistFlow approximations under uncertain network parameters. 
\end{enumerate}

In Section \ref{sec:ybus} we present concentration inequalities to control the spectrum of random admittance matrices under fixed and uncertain connectivity. We outline several applications of the proposed approach, such as bounding the error of linear power flow approximations under parameter uncertainty, including the LCPF formulation of the power flow equations, in Section \ref{sec:LCPF_bounds} and bounds for popular linear approximates relative to the AC power flow manifold in Section \ref{sec:grid_state}. Additionally, in Section \ref{sec:validation} we numerically validate an application on contingency analysis. We conclude in Section \ref{sec:conclusions} with limitations and future work.

\section{Bounding the Spectrum of Admittance Matrices}
\label{sec:ybus}
\newcommand{\mYbar}{\mtx{\bar{Y}}}
In this section, we present bounds on the error of the admittance matrix $\mY \in \C^{n \times n}$. In particular, we will study the following block Laplacian matrix operator, also known as the \textit{flat start Jacobian}. This matrix captures the spectral behavior of both the admittance matrix, and the behavior of linear approximations of the power flow equations about the flat start; see Appendix~\ref{apdx:lcpf} for details.
\begin{definition}[Flat start Jacobian]
    \label{def:pf_operator}
    For an arbitrary power network modeled by an undirected graph $\setG = (\setN,\setE)$, where $n := \abs{\setN}$ and $m := \abs{\setE}$, define the \emph{linear power flow operator} 
    \begin{equation}
        \mF := \begin{bmatrix}
            \mG & - \mB\\
            -\mB & -\mG
        \end{bmatrix}
        \in \R^{2n \times 2n},
    \end{equation} 
    where $\mG,\mB \in \R^{n \times n}$ are graph Laplacian matrices corresponding to networks with sufficiently many strictly positive (i.e., conductance) edge weights such that $\rank(\mG) \geq  n-1$ or real-valued (i.e., susceptance) edge weights, respectively. 
\end{definition} 

The matrix $\mF$ is the standard power flow Jacobian matrix evaluated at the flat start condition; see Appendix~\ref{apdx:lcpf} for a concise derivation. The matrix~$\mF$ is symmetric-indefinite, due to the following practical assumption.
\begin{assumption}
    The conductance and susceptance edges weights are such that
    \[
    \mG \succeq 0, \qquad \mathsf{and} \qquad \mB \preceq 0.
    \]
    See \cite{chen_definiteness_2016} for a discussion of the conditions on the susceptances required for $\mB \preceq 0$ to hold.
\end{assumption}
In addition to being the flat start Jacobian, the matrix~$\mF$ can also be interpreted as being~\emph{equivalent to the admittance matrix up to phase shifts}. By this, we mean that the matrix~$\mF$ is equivalent to the lifted, real-valued version of the standard admittance matrix~$\mY$. In particular, if we let~$\mG = \Re(\mY)$ and~$\mB = \Im(\mY)$, and define
\begin{equation}
    \label{eq:lifted_admittance_matrix}
    \mYbar := 
\begin{bmatrix}
    \mG & \mB\\
    \mB & -\mG
\end{bmatrix}
\in \R^{2n \times 2n},
\end{equation}
only the sign of~$\mB$ is flipped in~\eqref{eq:lifted_admittance_matrix} compared to~$\mF$, i.e.~$\mYbar$ is isomorphic to the Jacobian~$\mF$ under a phase shift of~$\pi$. 

Note that the matrix~$\mYbar$ is symmetric; it can be easily seen to have the same operator norm as~$\mY$. Its~$2 \times 2$ block structure suggests the decomposition as a sum of Kronecker products~\cite{horn2012matrix}:
\[
\begin{aligned}
\mYbar 
= \sum_{(i,j) \in \setE} 
\begin{bmatrix}
    g_{ij} & b_{ij}\\
    b_{ij} & -g_{ij}
\end{bmatrix}
\otimes \mE_{ij} 
&:= \sum_{(i,j) \in \setE} \mUpsilon_{ij} \otimes \mE_{ij}  \\
&:= 
\sum_{(i,j) \in \setE} \mM_{ij},
\end{aligned}
\]
where $\cb{\mUpsilon_{ij}}_{(i,j) \in \setE}$ is a sequence of $2 \times 2$ \textit{random symmetric matrices} representing the uncertainty in the connection $\left(i, j\right) \in \setE$, and the sequence of random matrices $\cb{\mM_{ij}}_{ij \in \setE}$ decompose the entire network in terms of elementary Laplacian matrices, which we now define.

\begin{definition}[Elementary Laplacian Matrix]
\label{def:elementary_laplacian}
 For each line~$(i,j) \in \setE$, let $\mE_{ij} \succeq 0$ denote the positive-semidefinite \emph{elementary Laplacian matrix}
\[
\mE_{ij}^{\vphantom\top} := \ve_{ij}^{\vphantom\top}\ve_{ij}^\T := (\ve_i - \ve_j)(\ve_i - \ve_j)^\T,
\]
describing the normalized subgraph between each pair $i,j$.

\end{definition}

\begin{figure}
    \centering
    \includegraphics[width=0.98\linewidth]{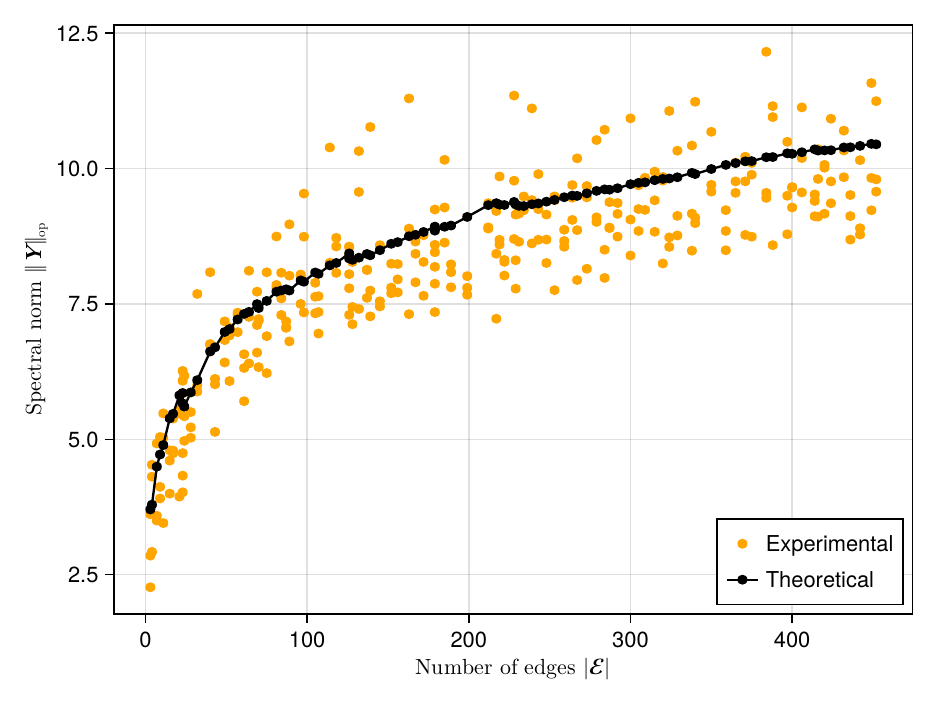}
    \caption{Comparison between the analytical bound for expected operator norm~$\E[\norm{\mY}]$ of the admittance matrix and 200 experimental samples, plotted against the number of lines in the network. In this simple numerical experiment, the networks were generated using the homogeneous Erd\HH{o}s-R\'enyi model, i.e.~by switching all possible lines independently with some probability~$p$, and changing~$p$ to increase the number of switched lines. Minor discontinuities in the theoretical curve are due to randomness in the number of switched lines.}
    \label{fig:enter-label}
\end{figure}

\subsection{An illustrative example: Admittances bounded by 1 per-unit}
\label{sec:illustrative}
In many applications, it is useful to understand how a power network will behave under uncertainty in the admittance parameters, where the uncertainty has a bounded size. A useful tool for analyzing matrices that have this property is the matrix Bernstein inequality.
\begin{theorem}[A matrix Bernstein bound \cite{tropp_concentration_2015}] \label{thm:Matrix_Bernstein}
Suppose that $\mX_1,\ldots,\mX_k \in \R^{n \times n}$ are independent, symmetric, zero-mean random matrices such that $\norm{\mX_i} \leq R$ always. Let $\mX := \sum_{i=1}^k \mX_i$ and define the matrix statistic \[
\nu\p{\mX} :=\norm{\var\p{\mX}} = \norm{\sum_{i=1}^k \Expec[]{\mX_i^2} }.
\]
Then
\[
\Expec[]{\norm{\mX}} \leq  \sqrt{2\nu\p{\mX}\log\p{2n}} + \frac{R}{3}\log\p{2n}
\]
Furthermore, for any $t>0$, we have
\[
\Pr\p{\norm{\mX} \geq t} \leq 2n \exp\p{\frac{-t^2/2}{\nu\p{\mX}+Rt/3}}.
\]
\end{theorem}

We now give an illustrative example of bounded admittance uncertainty in a network with fixed, known connectivity. In the sequel, we will allow for uncertain connectivity. Note that the same matrix Bernstein bound can be easily extended to the case of a sum of \emph{uncentered} random matrices (see Section 6.1 of \cite{tropp_concentration_2015}). The zero mean assumption is to simplify details rather than a fundamental limitation of our methodology. 


\begin{theorem}[Concentration of the admittance matrix with fixed connectivity and bounded admittances]
\label{thm:random_admittance_matrix}
    Consider a power system with $n$ nodes and $m$ lines. Let $\Delta = \max_{i} \deg(i)$ be the maximum degree of any node in the network. Suppose that the admittances are distributed according to any bounded distribution~$w_{i,j} \sim \setD$ that satisfies~$\abs{w_{i,j}}\leq 1$. Then, we have
    \[
    \Expec[]{\norm{\mY}} \leq \sqrt{4 \Delta\log(2n)} + \frac{2}{3}\log(2n).
    \]
\end{theorem}
The proof appears in Appendix~\ref{apdx:proof:random_admittance_matrix}.

\subsection{Uncertain contingencies}
\label{sec:bound_Y_conting}
In this section, we present bounds on the spectrum of the admittance matrix under random contingencies, shown in Theorem \ref{thm:conc_rand_conting}. This bound is useful for analyzing the behavior of the power flow equations under uncertain changes in network topology. This has natural applications in many relevant settings, for example, during natural disasters, public safety power shut-offs, or faults. Throughout this section, we operate under the following assumption.
\begin{assumption}[Uncertain contingencies]
\label{assum:rand_conting}
    Suppose that each line~$l=(i,j) \in \setE$ in a power network is switched closed (resp. switched open) with probability~$\pr_{l} \in [0,1]$ (resp.~$1-\pr_{l}$).
\end{assumption}
Assumption \ref{assum:rand_conting} is equivalent to modeling the power network as an \textit{inhomogeneous} Erd\HH{o}s-R\'enyi graph. It is highly relevant in the context of risk-based optimal transmission switching; see \cite{jeanson2025riskbasedapproachoptimaltransmission} for example.

We will analyze how the admittance matrix behaves under the setting of Assumption \ref{assum:rand_conting}. To achieve this, we will define the following notion of \textit{contingency factors}, and the \textit{criticality} of a node.
\begin{definition}[Contingency factors and nodal criticality]
\label{def:conting_factors}
    Consider a power network in the context of Assumption \ref{assum:rand_conting} with line admittances~$\cb{y_l}$. Define the \textit{contingency factors}~$\cb{c_l}$ of each line~$l=(i,j) \in \s{m}$ as
    \begin{equation}
    \label{eq:conting_fact}
    c_l := 2\cdot \pr_l\p{1-\pr_l}\abs{y_l}^2
    \end{equation}
    and the \textit{degree of criticality} of each node~$i$ under the contingency factors~$\vc \in \R^m_+$ is defined as
    \begin{equation}
    \label{eq:def:deg_crit}
    d_i(\vc) := \sum_{l:l \ni i} c_l  = \sum_{l : l \ni i} 2\cdot\pr_l\p{1-\pr_l}\abs{y_l}^2.
    \end{equation}
    Moreover, we denote the~\textit{maximum criticality} under~$\vc$ as
    \[
    \Delta_{\vc} := \max_{i \,\in\, [n]}\, d_i(\vc).
    \]
\end{definition}
In essence, the objects in Definition~\ref{def:conting_factors} are the edge weights, the nodal degrees, and the maximum nodal degree, respectively, of a certain graph Laplacian matrix. In particular, it is the Laplacian matrix that arises from the matrix-valued variance of the admittance matrix under uncertain contingencies, which we define explicitly in the forthcoming result.

\begin{theorem}[Concentration with fixed admittances and uncertain contingencies]
    \label{thm:conc_rand_conting}
    Consider a power network in the context of Definition \ref{def:conting_factors}. Let each line~$l=(i,j) \in \setE$ have admittance~$y_{ij} \in \C$ with~$\abs{y_{ij}} \leq 1$ per-unit.
    Define the random line edge weights
    \[
    w_{ij} := y_{ij}\cdot \xi_{ij}, \quad \xi_{ij}\sim\mathsf{Ber}(\pr_{ij}), \quad (i,j) \in E,
    \]
    and the corresponding random admittance matrix
    \begin{equation}
        \label{eq:rand_Y_conting}
        \mY := \sum_{(i,j) \in \setE} \xi_{ij}y_{ij}\p{\ve_i-\ve_j}\p{\ve_i-\ve_j}^\T\! 
    \end{equation}
    and center as~$\mYtilde := \mY - \E \mY$. 
    Define the normalized total degree of criticality:
    \begin{equation}
        \label{eq:norm_tot_deg_crit}
         \bar{D}:=\Delta_{\vc}^{-1}\sum_{i\in\s{n}} d_i(\vc).
    \end{equation}
    Then, for all~$t\geq\sqrt{2\Delta_{\vc}} + 2/3$, we have
    \begin{equation}
        \label{eq:tail_Y_conting}
        \Pr(\|\mYtilde\|\geq t) \leq 4 \bar{D}\cdot\exp\!\p{\frac{-t^2}{4\p{\Delta_{\vc}+t/3}}}; 
    \end{equation}
    moreover, there exists a constant~$C>0$ such that
    \begin{equation}
        \label{eq:expect_Y_conting}
        \E\|\mYtilde\| \leq C\p{\sqrt{2\Delta_{\vc}\log(1+\bar{D})} +2\log(1+\bar{D})}.
    \end{equation}
\end{theorem}
The proof of Theorem~\ref{thm:conc_rand_conting} appears in Appendix~\ref{apdx:proof:conc_rand_conting}.

\begin{remark}
    The matrix \eqref{eq:rand_Y_conting} is not positive-semidefinite, except in extremely restrictive cases; e.g., $\mG\mB=\mB\mG$ is one sufficient condition. 
    Hence, we must use the Bernstein inequality, as opposed to bounds with potentially simpler forms, like the matrix Chernoff inequality.
\end{remark}


\section{The Linear Coupled Power Flow Model: Formulation and bounds}\label{sec:LCPF_bounds}
The LinDistFlow equations are known to be equivalent to the Linear Coupled Power Flow (LCPF) model in the special case that the network is a tree. We will continue analysis working with the LCPF model for this section as it is more general and results for LinDistFlow follow simply. See Appendix \ref{apdx:lcpf} for more details.

We can decompose $\mF$, as defined in \eqref{eq:linear_coupled_power_flow} into the sum of Kronecker products between a particular $2 \times 2$ block matrix of admittances and elementary graph Laplacian matrices (Def. \ref{def:elementary_laplacian}), as follows:
\begin{align*}
    \mF &= \begin{bmatrix}
        \mA^\T \diag(\vg) \mA & -\mA^\T \diag(\vb) \mA\\
        -\mA^\T \diag(\vb) \mA & - \mA^\T \diag(\vg) \mA
    \end{bmatrix}\\
    &=\begin{bmatrix}
        \sum_{ij \in \setE} \mE_{ij} g_{ij} & -\sum_{ij \in \setE} \mE_{ij} b_{ij}\\
        -\sum_{ij \in \setE} \mE_{ij} b_{ij} & -\sum_{ij \in \setE} \mE_{ij} g_{ij}
    \end{bmatrix}\\
    &=\sum_{ij \in \setE} \begin{bmatrix}
        \mE_{ij} g_{ij} &  -\mE_{ij} b_{ij}\\
        - \mE_{ij} b_{ij} & - \mE_{ij} g_{ij}
    \end{bmatrix}\\
    &=\sum_{ij \in \setE} \underbrace{\begin{bmatrix}
        g_{ij} & -b_{ij}\\
        -b_{ij} & -g_{ij}
    \end{bmatrix}}_{:= \mUpsilon_{ij}} \otimes \mE_{ij}\\
    &= \sum_{ij \in \setE} \mUpsilon_{ij} \otimes \mE_{ij} := \sum_{ij \in \setE} \mM_{ij},
\end{align*}
where~$\mUpsilon_{ij}$ is a $2 \times 2$ symmetric matrix of admittances for a given line~$(i,j)$, as defined above.


\subsection{Bounded spectra of elementary power flow Jacobians}

Let $\mM_{ij} = \mUpsilon_{ij} \otimes \mE_{ij}$ be the elementary Jacobian corresponding to edge $\left(i, j\right) \in \setE$. Note that the operator norm of $\mM_{ij}$ is
\[
\norm{\mM_{ij}} \overset{(1)}{=} \norm{\mUpsilon_{ij}}\norm{\mE_{ij}} \overset{(2)}{=} 2\sqrt{g_{ij}^2+b_{ij}^2} = 2 \abs{y_{ij}}
\]
where step (1) is due to the fact that $\norm{\mA \otimes \mB}=\norm{\mA}\norm{\mB}$ for any matrices $\mA,\mB$, and step (2) is due to the fact that $\norm{\mE_{ij}} = 2$, and furthermore,
\begin{align*}
\norm{\mUpsilon_{ij}} &= \sqrt{\lambda_{\sf max}\p{\mUpsilon_{ij}^\T \mUpsilon_{ij}}}\\
&= \sqrt{\lambda_{\sf max}\p{\begin{bmatrix}
    g_{ij}^2 + b_{ij}^2 & 0\\
    0 & g_{ij}^2 + b_{ij}^2
\end{bmatrix}}}\\
&= \sqrt{g_{ij}^2 + b_{ij}^2} = \abs{y_{ij}}.    
\end{align*}

\begin{remark}
    Note that the following identities hold:
    \begin{equation}
    \norm{\mM_{ij}} = \sqrt{2\trace\p{\mUpsilon_{ij}^\T \mUpsilon_{ij}}} = \sqrt{2}\norm{\mUpsilon_{ij}}_F,
    \end{equation}
    and 
    $
    \norm{\mM_{ij}}_F = 2\sqrt{2}\norm{\mUpsilon_{ij}}.
    $
\end{remark}

\subsection{Matrix variance of the LCPF model under uncertainty}
\label{sec:prop:matrix_variance}
We now bound the matrix variance of the LCPF model. 
\begin{lemma}
    \label{lemma:matrix_variance}
    Suppose that $\vg,\vb \in \R^m$ are independent and uniformly distributed on $(m-1)$-dimensional sphere of radius~$1/2$; $\vg,\vb \overset{\sf iid}{\sim} \dUnif\p{\{\vy \in \R^n \: : \: \vy^\T \vy = \frac{1}{2}\}}$. Then the matrix-valued variance of the linear power flow operator (Def. \ref{def:pf_operator}) is upper-bounded as
    \[
    \Expec[]{\mF \mF^*} \preceq \frac{2}{n} \mId_2 \otimes \mA^\T \mA := \mV.
    \]
\end{lemma}
The proof appears in Appendix~\ref{apdx:proof:matrix_variance}.


\begin{theorem}[Spectral error of LCPF model under uncertain admittance parameters]
\label{thm:bounded-lcpf}
Suppose that the conductance and susceptance parameters of an electric power system model are uncertain, and can be modeled for each line $(i,j) \in \setE$ as 
\begin{subequations}
    \begin{align}
        g_{ij} = g^\bullet_{ij} + \Delta_{ij}^g\\
        b_{ij} = b^{\bullet}_{ij} + \Delta_{ij}^b,
    \end{align}
\end{subequations}
where $g^\bullet_{ij}, b^\bullet_{ij} \in \R$ are the true parameters, and $\Delta_{ij}^g$, $\Delta_{ij}^b$ are bounded uncertainty random variables such that 
\[
\max\cb{\norm{\Delta_{ij}^b},\norm{\Delta_{ij}^g}} \leq \Delta \quad \forall (i,j) \in \setE.
\]
Then, for any $t\geq 0$, we have that
\begin{equation}
\label{eq:prob_bound}
    \Pr\p{\norm{\mF - \Expec[]{\mF}} \geq t} \lesssim n \exp\p{\frac{-t^2}{4\p{\Delta^2 n + \Delta t/3}}},
\end{equation}
and moreover,
\begin{equation}
    \label{eq:expec_bound}
    \Expec[]{\norm{\mF - \Expec[]{\mF}}} \leq 2 \Delta \sqrt{2} \p{\sqrt{n \log\p{4n}} + \frac{1}{3} \log\p{4n}}
\end{equation}
\end{theorem}
The proof appears in Appendix~\ref{apdx:proof:bounded-lcpf}.

\section{Application: Bounding the Error of a Family of Power Flow Linearizations}
\label{sec:grid_state}
In this section, we provide an error bound of linear approximations of the AC power flow equations under uncertain admittances. This serves as a useful primitive for further applications on the evaluation of the quality of DC power flow in practical problems such as contingency analysis and network reconfiguration. Following the very recent results of \cite{goodwin_geometry_approx_2025}, we can utilize the perspective that the power flow equations admit a manifold interpretation to perform such analysis under uncertain admittances. Recounting the setup of \cite{goodwin_geometry_approx_2025}, let the AC power flow manifold be 
\[
\setM_{PF} = \mathrm{gph}(\Phi)= \cb{(\vx, \vp,\vq ) \in \R^{4n} : (\vp,\vq) = \Phi (x)},
\] 
where~$\vx = (\vv, \vtheta) \in \R^{2n}$, ~$\Phi : \R^{2n} \to \R^{2n}$ is the power flow mapping, and graph is defined as
\[
\mathrm{gph}(f) := \{ (\vx, f(\vx)) \in \R^{n+m} \mid \vx \in X \}
\]
where $f : X \to Y$, $X \subset \R^n$, and $Y \subset \R^m$.
Equivalently, ~$\setM_{PF} = F^{-1} (0)$ with~$F(\vx,\vp,\vq) = \Phi(\vx) -(\vp,\vq)$ and
\[
\mathsf{D}F(\vx,\vp,\vq) = \begin{bmatrix}\mathsf{D}\Phi_{\vx}& -I_{2n} \end{bmatrix}
\]
is surjective everywhere. Fix a feasible base point 
$$
\vz_\star = (\vx_\star, \vp_\star, \vq_\star) = (\vx_\star, \Phi(\vx_\star)) \in \setM_{PF}.
$$ 
For a step $\vh \in \R^{2n}$ tangent at $\vz_*$ (i.e. implicit function theorem evaluated at $\vz_*$ s.t. $(\vh, \mathsf{D}\Phi_{\vx_*} \vh) \in T_{\vz_*} \setM_{PF}$), form the tangent point $\bar{\vz} := \vz_* + (\vh, \mathsf{D}\Phi_{\vx_*}\vh)$, where $T_{\vz}\setM_{PF} = \ker\begin{bmatrix}\mathsf{D}\Phi_{\vx}& -I_{2n} \end{bmatrix}$ is the tangent space at $\vz$.
 \begin{proposition} \label{thm:manifold_error_bound}
     Let $\bar{\vz} \in T_{\vz} \setM_{PF}$ be a point in the linear tangent space of the power flow manifold about $\vz \in \setM_{PF}$. For a random admittance matrix $\mY$, defined in a similar manner as Theorem \ref{thm:random_admittance_matrix}, the expected Euclidean distance (in its ambient space) of a tangent step from a random AC PF manifold is
     \begin{align*}        
     \Expec[]{\mathrm{dist}(\bar{\vz}, \setM_{PF})} &\leq 3\|\vh\|_{\infty} \|\vh\|_2 \Expec[]{\norm{\mY}} \\
     & \leq 3 \| \vh\|_2^2 \left (\sqrt{8\Delta\log(4n)}+ \frac{2}{3}\log (4n)  \right).
     \end{align*}
 \end{proposition}

 \begin{proof}
     From \cite[Prop III.1]{goodwin_geometry_approx_2025}, we immediately get 
     \[
     \mathrm{dist}(\bar{\vz}, \setM_{PF}) \leq 3 \norm{F(\bar{\vz})}.
     \]
     By definition, $F(\vx,\vp,\vq) = \Phi(\vx) - (\vp,\vq).$ Therefore,
     \begin{align}\label{pf_taylor_expansion}
     F(\bar{\vz}) = \Phi(\vx_* + \vh) -(\Phi(\vx_*) + \mathsf{D}\Phi_{\vx_*}\vh)
     \end{align}
     The AC PF manifold can be equivalently represented as 
     \[
     \setM_{PF}= \{(\vu,\vs) \in \C^n \times \C^n : \vs = \Psi_{\mY} (\vu):= \diag(\vu) \conj{\mY \vu} \}
     \]
     where $\conj{(\cdot)}$ denotes complex conjugate. Let $(\vu_*, \vs_*)$ be the complex form of $\vz_*$, so the complex tangent point is
     \[
     (\bar{\vu}, \bar{\vs}) : = (\vu_* + \vh_u, \vs+ \mathsf{D}\Psi_{\mY} (\vu_*) [\vh_u])
     \]
     where $\vh_u \in \C^n$ is the complex voltage step. Since $\Psi_{\mY}$ is quadratic in $\vu$,
     \begin{align*}
     \Psi_{\mY} (\vu_* + \vh_u) = \Psi_{\mY}(\vu_*) &+ \mathsf{D}\Psi_{\mY} (\vu_*) [\vh_u]\\ &+ \tfrac{1}{2} \mathsf{D}^2\Psi_{\mY} (\vu_*) [\vh_u, \vh_u]     
     \end{align*}
     with first and second Fr\'echet derivatives as
     \[
     \mathsf{D}\Psi_{\mY} (\vu) [\vh_u] = \diag(\vh_u)\conj{\mY \vu} + \diag(\vu) \conj{\mY \vh_u},
     \]
     \[
     \mathsf{D}^2\Psi_{\mY} [\vh_u, \vk_u] = \diag(\vh_u)\conj{\mY \vk_u} + \diag(\vk_u) \conj{\mY \vh_u}.
     \]
     Subtracting the linearization,
     \begin{align*}
     \Psi_{\mY}(\vu_* + \vh_u) -&(\Psi(\vu_*) + \mathsf{D}\Psi_{\vx_*}[\vh_u]) \\ &= \frac{1}{2}\mathsf{D}^2\Psi_{\mY} (\vu_*) [\vh_u, \vh_u] \\
     &= \diag(\vh_u)\conj{\mY \vh_u}
     \end{align*}
     Since this is the same form as \eqref{pf_taylor_expansion} and since $\C^n \cong \R^{2n}$,
     \[
     \| F(\bar{\vz}) \|_2 = \| \diag(\vh_u) \conj{\mY \vh_u}\|_2
     \]
     and by standard norm inequalities, we have
     \begin{align*}
     \| F(\bar{\vz}) \|_2 &= \| \diag(\vh_u) \conj{\mY \vh_u}\|_2 \\
     &\leq \norm{\diag(\vh)} \norm{\mY} \|\vh\|_2 = \|\vh\|_{\infty} \norm{\mY} \|\vh\|_2 \\
     &\leq \|\vh\|_2^2 \norm{\mY}.
     \end{align*}
     Plugging into \cite[Prop III.1]{goodwin_geometry_approx_2025}
     \begin{align*}
     \mathrm{dist}(\bar{\vz}, \setM_{PF}) &\leq 3 \|F(\bar{\vz})\| \\
     &\leq 3\|\vh\|_{\infty} \norm{\mY} \|\vh\|_2 \\
     &\leq 3\|\vh\|_2^2 \norm{\mY}
     \end{align*}
     Taking expectations on both sides and using Thm \ref{thm:random_admittance_matrix} yields
     \begin{align*}
        \Expec[]{\mathrm{dist}(\bar{\vz}, \setM_{PF})} &\leq 3\|\vh\|_{\infty} \|\vh\|_2 \Expec[]{\norm{\mY}} \\
        &\leq 3 \|\vh\|_2^2 \left( \sqrt{8 \Delta\log(4n)} + \frac{2}{3}\log(4n)\right ).
     \end{align*} 
 \end{proof}

There is a useful interpretation for the expression derived in Prop \ref{thm:manifold_error_bound}. It is an upper bound on the expected worst case deviation of a local tangent step from a given point of a \emph{typical} ACPF manifold. We are able to describe a family of corresponding manifolds generated by this random admittance matrix that corresponding to a large class of physically realizable grids because we only assume the line admittances come from any bounded probability distribution. Our result also suggests that, in a sense, the geometry of the ACPF manifold itself concentrates in a manner described in Thm \ref{thm:random_admittance_matrix}. 

\begin{proposition} \label{thm:manifold_special_case}
    Suppose the special case where we assume a lossless network (i.e. $\mY = -j\mB$) with the susceptances are random variables of the form $w_e = b_e \cdot s_e$ where $b_e$ is the physical susceptance and $ s_e \sim \mathsf{Ber}(p_e)$, where $p_e$ is the rate at which line $e$ is switched closed. From some constant $C>0$, the expected distance from a local point on the linear tangent space to the random AC PF manifold is 
    \begin{align*}
        \Expec[]{\mathrm{dist}(\bar{\vz}, \setM_{PF})} \leq \, &3C\|\vh\|_{\infty} \|\vh\|_2 \\
        &\p{\sqrt{2\Delta_{\vc}\log(1+\bar{D})} +2\log(1+\bar{D})}
    \end{align*}
\end{proposition}

\begin{proof}
    Following directly from Prop \ref{thm:manifold_error_bound}, we know that 
    \[
    \Expec[]{\mathrm{dist}(\bar{\vz}, \setM_{PF})} \leq  3\|\vh\|_{\infty} \|\vh\|_2 \Expec[]{\norm{\mY}}
    \]
    so what's left is to bound this special case of $\Expec[]{\norm{\mY}}$. \\ By lossless assumption and norm properties, $$\norm{\mY} = \norm{-j\mB} = \norm{\mB}.$$
    From Theorem $\ref{thm:conc_rand_conting}$, we directly get that 
    \[
    \E\|\mB\| \leq C\p{\sqrt{2\Delta_{\vc}\log(1+\bar{D})} +2\log(1+\bar{D})},
    \]
    so the expected distance from the ACPF manifold is
    \begin{align*}
        \Expec[]{\mathrm{dist}(\bar{\vz}, \setM_{PF})} \leq \, &3C\|\vh\|_{\infty} \|\vh\|_2 \\
        &\p{\sqrt{2\Delta_{\vc}\log(1+\bar{D})} +2\log(1+\bar{D})}
    \end{align*}
    for some constant $C>0$.
\end{proof}

\begin{remark}
    Remember the error bounds presented in Prop \ref{thm:manifold_error_bound} and \ref{thm:manifold_special_case} are measuring a distance of a vector $\bar{\vz}$ between a linearization to the ACPF manifold in its implicit/ambient space, i.e. the space of all vectors $(\vv, \vtheta, \vp, \vq) \in \R^{4n}$. It may be of more direct interest in applications to consider this distance with respect to a specific quantity of interest, such as voltage magnitude. In these settings, it is sufficient to consider a projection of the vector $\bar{\vz}$ to that quantity of interest. We defer proof for future work, but the distance of this projected vector will have the same scaling as the distance of $\bar{\vz}$, only differing by a constant.   
\end{remark}

\section{Numerical Validation on IEEE Networks}
\label{sec:validation}

We validate our theoretical bounds on standard IEEE test networks and synthetic radial topologies using Monte Carlo simulations. For each network and switching probability $p \in \{0.1, 0.3, 0.5, 0.7, 0.9\}$, we generate 500 random contingency samples and compute the empirical distribution of $\|\mYtilde\| = \|\mY - \E\mY\|$.

\subsection{Bound Evaluation}
 We experimentally evaluate the expectation bound from Thm \ref{thm:conc_rand_conting} against the empirical distribution and a sharper bound from~\cite{bandeira_sharp_2016} but assumes independent entries of our random matrix.

Substituting $\nu = 2\Delta_{\vc}$ and $L=2$ the expectation bound from Theorem~\ref{thm:conc_rand_conting} gives:
\begin{equation}
\label{eq:tropp_validation}
\Expec[]{\|\mYtilde\|} \leq \sqrt{4\Delta_{\vc}\log(1+d)} + \tfrac{2}{3}\log(1+d).
\end{equation}

In ~\cite{bandeira_sharp_2016}, it states that for a symmetric random matrix $\mX$ with independent entries,
\begin{equation}
\label{eq:vanhandel_bound}
\Expec[]{\|\mX\|} \lesssim \sigma_{\mathsf{row}} + \sigma_{\mathsf{max}}\sqrt{\log n},
\end{equation}
where $\sigma_{\mathsf{row}} := \max_i \sqrt{\sum_j \Var{\mX_{ij}}}$ is the maximum row standard deviation and $\sigma_{\mathsf{max}} := \max_{ij} \sqrt{\Var{\mX_{ij}}}$ is the maximum entry standard deviation. For our centered admittance matrix $\mYtilde$, we have $\sigma_{\mathsf{row}} = \sqrt{\Delta_{\vc}}$. The key advantage of~\eqref{eq:vanhandel_bound} is that the leading term $\sigma_{\mathsf{row}} = \sqrt{\Delta_{\vc}}$ has no $\log$ factor, unlike~\eqref{eq:tropp_validation}. Note that the admittance matrix, and graph Laplacian matrices in general, trivially do not satisfy the requirements of the bound in \eqref{eq:vanhandel_bound}. However, in section \ref{subsec:discussions} we discuss why this bound is of potential interest and how it can make sense of the results. 
\begin{table}[h]
    \centering
    \caption{Validation results at $p=0.5$ (maximum expected norm). Bound ratio = empirical mean / theoretical bound.}
    \label{tab:validation}
    \begin{tabular}{lccccc}
        \hline
        Network & $n$ & $m$ & Empirical & \eqref{eq:tropp_validation} ratio & \eqref{eq:vanhandel_bound} ratio \\
        \hline
        IEEE 14 & 14 & 20 & $24.0 \pm 2.8$ & 0.44 & 0.67 \\
        IEEE 30 & 30 & 41 & $40.9 \pm 4.1$ & 0.40 & 0.65 \\
        \hline
    \end{tabular}
\end{table}

Figure~\ref{fig:validation} shows that both bounds correctly bound the empirical mean. For the~\eqref{eq:tropp_validation} bound, we get a ratio of ${\sim}0.40$--$0.44$, i.e., ${\sim}2$--$2.5\times$ conservative. For the~\eqref{eq:vanhandel_bound} bound, the ratio is ${\sim}0.65$--$0.70$, i.e., ${\sim}1.5\times$ conservative.
The~\eqref{eq:vanhandel_bound} bound is tighter because its leading term $\sqrt{\Delta_{\vc}}$ has no $\log$ factor, while the~\eqref{eq:tropp_validation} bound has $\sqrt{\Delta_{\vc} \log d}$.

\begin{figure}[t]
    \centering
    \includegraphics[width=0.98\linewidth]{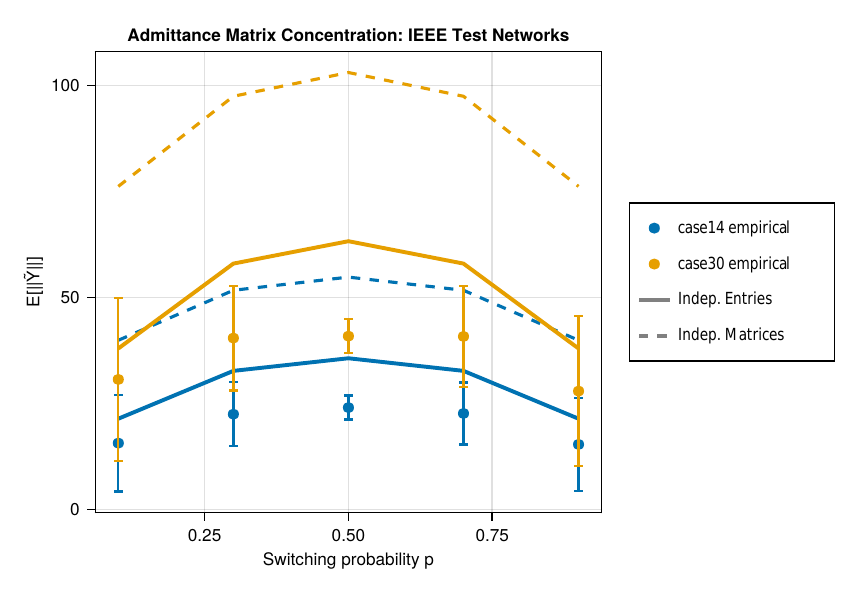}
    \caption{Comparison of empirical $\E\|\mYtilde\|$ (markers) versus theoretical bounds across switching probabilities on IEEE test networks. Error bars show $\pm 1$ sample standard deviation over 500 Monte Carlo samples. Solid lines: \cite{bandeira_sharp_2016} bound~\eqref{eq:vanhandel_bound}. Dashed lines: Theorem \ref{thm:conc_rand_conting} bound~\eqref{eq:tropp_validation}. Both bounds correctly bound the empirical mean and capture the parabolic shape in $p$ (maximum at $p=0.5$). Sample variability is highest at extreme $p$ (0.1, 0.9), reflecting the asymmetric distribution of the centered Bernoulli $\xi_l - p_l$.}
    \label{fig:validation}
\end{figure}

\subsection{Discussion}\label{subsec:discussions}
The numerical experiments confirm that our bounds:
\begin{enumerate}
    \item \textit{Bound the expectation:} The empirical mean $\E\|\mYtilde\|$ lies below both theoretical bounds in all configurations. Individual samples \emph{can} exceed the expectation bound (as visible in error bars), which is expected since these are bounds on the mean, not individual realizations.
    \item \textit{Capture scaling:} Both bounds correctly capture the parabolic shape in $p$ (maximum at $p=0.5$) and the $O(\sqrt{\Delta_{\vc}})$ scaling with network parameters.
\end{enumerate}
However, the~\eqref{eq:tropp_validation} bound developed in Theorem \ref{thm:conc_rand_conting} seems to give empirically looser bounds than the \eqref{eq:vanhandel_bound} bound from \cite{bandeira_sharp_2016}; why? There may be a sense that despite the tighter bound requiring independent entries of our random matrix, it is actually the right scaling for our structured problem. A key insight is that any looseness from invoking matrix Bernstein to solicit concentration results of structured random matrices may be attributed to the lack of exploiting noncommunicativity. In recent years, there has been advancement in the use of \emph{free probability theory} to develop sharpened concentration inequalities for random matrices utilizing ``intrinsic freeness." This may allow us to argue that our structured random admittance matrix can have the same scaling as a random matrix with independent entries, see \cite{HandelFreeProbability_2023, HandelFreeProbabilityII_2025}. We leave this direction of research for future work. 
For tighter estimates, the bound \eqref{eq:vanhandel_bound} from~\cite{bandeira_sharp_2016} appears numerically preferable, despite its theoretical difference from \eqref{eq:tropp_validation}. For bounding individual realizations with high probability, use the tail bound~\eqref{eq:tail_Y_conting}.





\section{Conclusions}\label{sec:conclusions}

In this paper, we presented a number of results applying matrix concentration inequalities to characterize behaviors of the power flow equations under uncertain admittances. We first derive an expectation bound on the spectrum of the admittance matrix under general distribution assumptions that scales on the network's maximum degree, and use it to develop refined tail bounds of uncertain contingencies expressed through contingency factors and nodal criticality. We then lift these results to the linear coupled power flow (LCPF) operators and show how the induced spectral uncertainty to linear approximations of the AC power flow manifold, producing explicit error bounds for a family of linear power flow models such as DC power flow. 

More precisely, our results imply that the expected operator norm of a random admittance matrix under general distributions of bounded admittances grows like $O(\sqrt{\Delta\log n} + \log n)$, where $\Delta$ is the maximum degree and $n$ is the numbers of nodes. This quantity has the interpretation that the under light distributional assumptions on the admittances, the effective resistance of a power network concentrates. From a modeling standpoint, controlling spectral uncertainty gives us access to provide guarantees of models (such as linearizations) that are related to uncertainty in topology/network parameters which are quantities we care about. Additionally, there are further applications these results support and leave for future work.
\subsection{Applications and Future Work}
\subsubsection{Contingency Analysis}
With a suitable Bernoulli parameter on the admittances, analyzing the distribution of power networks subsumes all possible $n-1$ configurations of the network. In particular, one can show that the simplex of all~$n-1$ contingencies does not violate any line flow constraint. For example, in the context of the DC approximation, the results of the present paper can show that
\[
\Pr\!\p{\norm{\vf}_\infty > \epsilon} \leq \delta(\epsilon),
\]
for an appropriate choice of~$(\epsilon,\delta)$. Here,
\[
\vf = \mS\mW\mA\vtheta, \qquad \vs \in \Delta_{n-1}^n
\]
with~$\mS:=\diag(\vs)$ and~$\mW=\diag(\vw)$ is a random vector of approximate line flows with~$\Delta_{n-1}^n$ denotes the set of all~$n-1$ contingency switching vectors.

\subsubsection{Network Reconfiguration}
Network reconfiguration is of great interest in recent power systems research, particularly for modern applications, such as in congestion management and grid planning. The combinatorial solution space of such problems implies that analyzing or sampling from a family of possible network configurations is a potentially relevant task. The theory of the present paper is directly applicable to such a setting, as it describes the behavior of such non-deterministic power flow models.

\subsubsection{Evaluation of Linearizations in Specialized Problems}
Following up on the results on the error bounds developed in this paper, the evaluation of linearizations such as DC power flow in more specific problems such as contingency analysis and network reconfiguration are of great interest to practitioners. This probabilistic framework could allow screening of linearization choices conformal to specific parameters of the problem. Moreover, applications involving bounding the error of locational marginal prices (LMPs) are also tractable, leveraging results that connect admittance matrices to LMP sensitivities~\cite{conejo_lmp_sensitivities_2005,Kekatos_2016}.


\section*{AI Usage Disclosure}
Claude Code with the Opus 4.5 model was used to help create a plotting interface to visualize the inequalities.

\appendix

\subsection{Fresh derivation of the Linear Coupled Power Flow Model}
\label{apdx:lcpf}

\begin{theorem}[Linear Coupled Power Flow Model \cite{bolognani_implicit_linear_2015,dhople_rectangular_linear_2015,deka_structure_learning_2018}]
\label{thm:linear_power_flow_model}
    Consider the \textit{flat start condition} $\vu_\star := \vone + \j \vzero$, and suppose that $\vomega=\vzero + \j \vzero$.
       Then, the linear coupled power flow manifold around $\vu_\star$ is the linear space
       \begin{equation}
           \setM_\star := \cb{\vx \in \R^{4n} : \mF(\vx_\star)(\vx - \vx_\star) = \vzero_{2n}},
       \end{equation}
       where $\mF : \R^{4n} \to \R^{2n \times 4n}$ is the Jacobian of the power flow equations $\loF$ at the nominal state $\vx_\star$, which we write as
       \begin{equation}\label{eq:linear_coupled_power_flow}
           \mF(\vx_\star) = \begin{bmatrix}
               \mG & - \mB & - \mId_{n \times n} & \vzero_{n \times n}\\
               -\mB & -\mG & \vzero_{n \times n} & - \mId_{n \times n}
           \end{bmatrix} = \begin{bmatrix}
               \mM & - \mId_{2n \times 2n}
           \end{bmatrix}.
       \end{equation}
       We define the $2n \times 2n$ matrix $\mM$ as the \emph{linear power flow matrix}.
       This matrix defines the linear power flow model
    \begin{equation}
    \boxed{
    \begin{bmatrix}
        \vp\\
        \vq
    \end{bmatrix} = \begin{bmatrix}
        \mG & - \mB\\
        -\mB & -\mG
    \end{bmatrix}
    \begin{bmatrix}
        \vepsilon\\
        \vtheta
    \end{bmatrix}
    }
    \end{equation}
    where $\vepsilon := \vv - \vone$.
    If the network is a tree with $n$ non-reference nodes and $n$ edges, the inverse of the linear power flow matrix $\mM$ is given in closed form as
    \begin{equation}
    \label{eq:thm:M_matrix_inv}
    \mM^{-1} := 
        \begin{bmatrix}
            \mG & - \mB\\
            -\mB& - \mG
        \end{bmatrix}^{-1} = \begin{bmatrix}
            \mR & \mX\\
            \mX & - \mR
        \end{bmatrix},
    \end{equation}
    where $\mR,\mX \succ 0$ are resistance and reactance matrices. Thus, 
    \begin{equation}
       \boxed{
       \begin{bmatrix}
           \vepsilon\\
           \vtheta
       \end{bmatrix} = \begin{bmatrix}
           \mR & \mX\\
           \mX & - \mR
       \end{bmatrix}\begin{bmatrix}
           \vp\\
           \vq
       \end{bmatrix},
       }
    \end{equation}
    where $\vepsilon := \vv - \vone$. 
\end{theorem}
\begin{proof}
The \emph{linear manifold tangent} to $\setM$ at a nominal operating point $\vx_\bullet$ is given by
    \begin{equation}
    \setM_\bullet := \cb{\vx \in \R^{4n} : \mF(\vx_\bullet)\p{\vx - \vx_\bullet} = \vzero_{2n}},
    \end{equation}
    where
    \begin{subequations}
    \label{eq:linear_approximant}
    \begin{align}
        \mF(\vx_\bullet) 
        &=
        \begin{bmatrix}
            \frac{\partial \loF}{\partial \vv}(\vx_\bullet) & \frac{\partial \loF}{\partial \vtheta}(\vx_\bullet) & \frac{\partial \loF}{\partial \vp}(\vx_\bullet) & \frac{\partial \loF}{\partial \vq}(\vx_\bullet)
        \end{bmatrix}
        \\
        &=
        \begin{bmatrix}
            \Re{\frac{\partial \vs}{\partial \vv}\p{\vu_\bullet}} & \Re{\frac{\partial \vs}{\partial \vtheta}\p{\vu_\bullet}} & -\mId_{n}&\vzero_{n}\\
            \Im{\frac{\partial \vs}{\partial \vv}\p{\vu_\bullet}} & \Im{\frac{\partial \vs}{\partial \vtheta}\p{\vu_\bullet}} & \vzero_{n} & -\mId_{n}
        \end{bmatrix}
        \\
        &=\begin{bmatrix}
                \frac{\partial \vp}{\partial \vv}(\vu_\bullet) & \frac{\partial \vp}{\partial \vtheta}(\vu_\bullet) & - \mId_{n}&\vzero_{n}\\
                \frac{\partial \vq}{\partial \vv}(\vu_\bullet) & \frac{\partial \vq}{\partial \vtheta}(\vu_\bullet) & \vzero_{n} & -\mId_{n}
            \end{bmatrix}.
    \end{align}
    \end{subequations}
    Let $\vomega := \vgamma + \j \vbeta \in \C^n$ denote the vector of self-admittances of each node. Then, following \cite[5.10]{molzahn_hiskens-fnt2019}, the Jacobian of complex power injections with respect to voltage phase angles is
    \begin{align*}
        \frac{\partial \vs}{\partial \vtheta}\p{\vu_\star} &=\j\diag\p{\vu_\star}\p{\diag\p{\conj{\mY\vu_\star}} - \conj{\mY}\diag\p{\conj{\vu_\star}}} \\
        &= \j  \mId_{n \times n}\p{\diag\p{\conj{\mY}\vone_n}  - \conj{\mY}\mId_{n \times n}}\\
        &=\j\p{\diag\p{\conj{\vomega}} -\conj{\mY}};
    \end{align*}
    similarly, with respect to the voltage magnitudes
    \begin{align*}
        \frac{\partial \vs}{\partial \vv}(\vu_\star) &= \diag\p{\vu}\p{\diag\p{\conj{\mY\vu}} + \conj{\mY}\diag\p{\conj{\vu}} }\diag\p{\vv}^{-1} \\
        &=\mId_{n \times n}\p{\diag\p{\conj{\mY }\vone_n} + \conj{\mY}\mId_{n \times n}} \mId_{n \times n}^{-1}\\
        &=\diag\p{\conj{\vomega}} + \conj{\mY}.
    \end{align*}
    Substituting the above into \eqref{eq:linear_approximant} yields the desired result.

    Now, we show that the assumption that the network is a tree, or radial, ensures that the linear power flow matrix $\mM$ can be inverted in the analytical form given in \eqref{eq:thm:M_matrix_inv}. Note that as $\mG \succ 0$, both Schur complements of $\mM$ exist. 

    Moreover, setting $\mS$ to be the Schur Complement of $\mM$ in $-\mG$, we obtain that the inverse of $\mS$  is, in fact, the \emph{resistance matrix} $\mR$, since
\begin{align*}
    \mS^{-1} &:= \p{\mG + \mB \mG^{-1} \mB}^{-1}\\
    &= \p{\mA^\T \diag(\vg) \mA + \mA^\T \diag(\vb) \diag(\vg)^{-1} \diag(\vb) \mA}^{-1}\\
    &=\p{\mA^\T \diag\p{ \s{\frac{g_{ij}^2+b_{ij}^2}{g_{ij}}}_{ij \in \setE} } \mA }^{-1}\\
    &=\mA^{-1} \diag\p{ \s{\frac{g_{ij}}{g_{ij}^2+b_{ij}^2}}_{ij \in \setE} } \mA^{-\T}\\
    &:=\mA^{-1}\diag\p{\vr} \mA^{-\T}\\
    &:= \mR.
\end{align*}
Further calculation reveals that the Schur complement of $\mM$ in $\mG$ is $-\mS$.

Therefore, applying well-known block matrix inversion identities yields
\begin{subequations}
    \label{eq:thm:block_M_inversion}
    \begin{align}
    \mF^{-1} &=    
    \begin{bmatrix}
        \mG & -\mB\\
        -\mB & - \mG
    \end{bmatrix}^{-1}\\
    &= \begin{bmatrix}
        \mS^{-1} & \vzero_{n \times n}\\
        \vzero_{n \times n} & -\mS^{-1}
    \end{bmatrix}\begin{bmatrix}
        \mId_{n \times n} & - \mB \mG^{-1}\\
        \mB \mG^{-1} & \mId_{n \times n}
    \end{bmatrix}
    \\
    &=\begin{bmatrix}
        \mS^{-1} & -\mS^{-1} \mB \mG^{-1}\\
        -\mS^{-1} \mB \mG^{-1} & -\mS^{-1}
    \end{bmatrix}.
    \end{align}
\end{subequations}
Finally, we  must compute the off-diagonal matrices of \eqref{eq:thm:block_M_inversion}. We obtain 
   \begin{align*}
        -\mS^{-1}\mB \mG^{-1} &= -\mA^{-1} \diag\p{ \s{\frac{g_{ij}}{g_{ij}^2+b_{ij}^2}}_{ij} } \diag(\vb \oslash \vg)\mA^{-\T}\\
        &=\mA^{-1} \diag\p{ \s{\frac{-b_{ij}}{g_{ij}^2+b_{ij}^2}}_{ij \in \setE} } \mA^{-\T}\\
        &:=\mA^{-1} \diag\p{\vx}\mA^{-\T}\\
        &:=\mX,
    \end{align*}
    where $\oslash$ denotes element-wise division.
    Therefore, we have that the inverse of the linear power flow matrix $\mM$ is
    \begin{subequations}
    \label{eq:thm:matrix_inverse_qed}
    \begin{align}
    \begin{bmatrix}
        \mG & - \mB\\
        -\mB & - \mG
    \end{bmatrix}^{-1} 
    &= \begin{bmatrix}
        \mS^{-1} & -\mS^{-1} \mB \mG^{-1}\\
        - \mS^{-1} \mB \mG^{-1} & - \mS^{-1}
    \end{bmatrix}\\
    &=\begin{bmatrix}
        \mR & \mX\\
        \mX & - \mR
    \end{bmatrix},    
    \end{align}
    \end{subequations}
    as desired.
\end{proof}

\subsection{Proof of Theorem \ref{thm:random_admittance_matrix}}
\label{apdx:proof:random_admittance_matrix}
\begin{proof}[Proof of Theorem~\ref{thm:random_admittance_matrix}]
First, bounding the operator norm uniformly across~$\mM_{ij}$, we have
\begin{align*}
    \norm{\mM_{ij}} &=\norm{\mUpsilon_{ij}}\cdot\norm{\mE_{ij}}\\
    &=2\sqrt{\lambda_{\max}\big(\mUpsilon_{ij}^\T \mUpsilon_{ij}^{\vphantom\T}}\big)
    =2\sqrt{g_{ij}^2+b_{ij}^2}\\
    &\leq 2\sup_{(i,j) \in \setE} \abs{w_{ij}} \leq 2 := R.
\end{align*}
On the other hand, denoting by~$\mZ^*$ the conjugate transpose of a matrix~$\mZ$, the matrix variance statistic~$\nu\p{\mYbar} := \big\|\E[{\mYbar}^2]\big\|
$ (see~\cite{tropp_concentration_2015}) can be expressed as follows:
\begin{equation}
\label{eq:chain-for-variance}
\begin{aligned}
    \nu\p{\mYbar} 
    &=\norm{\sum_{(i,j) \in \setE}\Expec[]{\p{\mUpsilon_{ij}^{\vphantom\T} \otimes \mE_{ij}^{\vphantom\T}} \p{\mUpsilon_{ij}^\T \otimes \mE_{ij}^\T} }}\\
    &\overset{(1)}{=}\norm{\sum_{(i,j) \in \setE} \Expec[]{\p{\mUpsilon_{ij}^{\vphantom\T}\mUpsilon_{ij}^\T} \otimes \p{\mE_{ij}\mE_{ij}^\T} }}\\
    &=2\norm{\sum_{(i,j) \in \setE} \begin{bmatrix}
        \mE_{ij} & 0\\
        0 & \mE_{ij}
    \end{bmatrix}}\\
    &=2\norm{\begin{bmatrix}
        \mA^\T\mA&\vzero\\
        \vzero&\mA^\T\mA
    \end{bmatrix}}\\
    &\overset{(2)}{=} 2 \big\|\mA^\T \mA\big\|\\
    &\overset{(3)}{\leq} 4\Delta.
\end{aligned}
\end{equation}
In this chain of equalities, step (1) is by the mixed-product property of the Kronecker product, namely,~$\p{\mA \otimes \mB}\p{\mC\otimes\mD} = \p{\mA \mC}\otimes\p{\mB\mD}$ for any matrices~$\mA,\mB,\mC,\mD$ with appropriate dimensions; step (2) follows since the operator norm of a block-diagonal matrix is the largest operator norm of any block.


For step (3), we observe that $\mA^\T \mA \in \R^{n \times n}$ is the graph Laplacian matrix of the simple undirected graph corresponding to the network topology. 
We use that~$\mY = \mD - \mM$ where~$\mM$ is the adjacency matrix; since~$\mY \succeq 0$ and~$\mD + \mM \succeq 0$ for the ``signless'' Laplacian matrix, see e.g.~\cite{cvetkovic2007signless}, we conclude that
\[
-\mD \preceq \mM \preceq \mD, \;\; \text{thus}\;\; \|\mM\| \le \|\mD\|
\]
and therefore~$\|\mY\| \le \|\mD\| + \|\mM\| \le 2\|\mD\| = 2\Delta$. A direct application of matrix Bernstein (Thm \ref{thm:Matrix_Bernstein}) gives the desired result.

\end{proof}

\subsection{Proof of Theorem \ref{thm:conc_rand_conting}}
\label{apdx:proof:conc_rand_conting}
\begin{proof}
    For each line~$l:=(i,j)\in \setE$, let~$\mM_{l} = \xi_{l}y_{l}\va_l\va_l^\T$ denote the summand matrices associated with each line in the matrix series~\eqref{eq:rand_Y_conting}. Observe that
    \[
    \E \mM_{l} = \pr_{l}y_{l}\va_l\va_l^\T
    \]
    is the expectation of each elementary admittance matrix. With this, define the \textit{centered} elementary admittance matrices as
    \[
    \mMtilde_{l} := \mM_l - \E \mM_l=\p{\xi_l-\pr_l}y_l\va_l\va_l^\T.
    \]
    The centered admittance matrix of the network is then
    \[
    \mYtilde = \mY-\E\mY =\sum_{l \in E} \mMtilde_l=\sum_{l \in E} \p{\xi_l-\pr_l}y_l\va_l\va_l^\T.
    \]
    Naturally, we have that~$\E \mMtilde_l = \vzero$ for any~$l$, and~$\E \mYtilde = \vzero$.

    Furthermore, for each line~$l$, we have the upper bound
    \begin{align*}
    \|\mMtilde_l\| &=\norm{\p{\xi_l - \pr_l}y_l\va_l\va_l^\T}\\
    &\overset{(1)}{=}\abs{\xi_l - \pr_l}\cdot\abs{y_l}\opnorm{\va_l\va_l^\T}\\
    &\overset{(2)}{\leq} 2,
    \end{align*}
    where step (1) is by absolute homogeneity and step (2) is due to the fact that~$\norm{\va_l\va_l^\T}=\norm{\va_l}_2^2=2$,~$\abs{y}_l \leq 1$ and
    $$
    \abs{\xi_l-\pr_l} \leq \max_{l}\,\abs{\xi_l-\pr_l} \leq \max_l\cb{\max\cb{\pr_l,1-\pr_l}}\leq 1.
    $$
    Now, we compute the matrix-valued variance of the centered admittance matrix under random contingencies. We have
    \begin{align*}
        \mV:=\E \mYtilde\mYtilde^* &\overset{(1)}{=} \sum_{l \in \setE} \E \mMtilde_l\mMtilde_l^*\\
        &=\sum_{l \in \setE} \Expec[]{\p{\xi_l-\pr_l}^2 \abs{y_l}^2 \va_l\va_l^\T \va_l\va_l^\T}\\
        &\overset{(2)}{=}\sum_{l \in \setE} 2\pr_l\p{1-\pr_l}\abs{y_l}^2 \va_l\va_l^\T.
    \end{align*}
    In the above display, step (1) is by independence of the summands, and step (2) is by definition of the Bernoulli variance~$\E\p{\xi_l-\pr_l}^2 = \pr_l\p{1-\pr_l}$, and~$\va_l^\T \va_l=2$. Observe that the matrix-valued variance~$\mV$ is itself a graph Laplacian matrix that describes a graph with the same topology as the power network, with the contingency factors as line weights. This Laplacian can be written as
    \[
    \mL  := \mA^\T \mC \mA, \qquad \mC = \diag(\vc).
    \]

    Now, we compute the intrinsic dimension of the matrix-valued variance, which is defined as follows. 
    \begin{definition}[Intrinsic dimension]
    \label{def:intdim}
    For any matrix $\mA$, let 
    \[
    \intdim(\mA) := \frac{\trace(\mA)}{\norm{\mA}}.
    \]
    \end{definition}
    First, note that we have
    \begin{align*}
        \intdim\!\p{\mV} &:= \frac{\trace\!\p{\mV}}{\norm{\mV}}
        =\frac{2\sum_{l \in \setE} \pr_l\p{1-\pr_l} \abs{y_l}^2\trace\!\p{\va_l\va_l^\T}}{\norm{\mV}}\\
        &=\frac{4\sum_{l \in \setE} \pr_l \p{1-\pr_l} \abs{y_l}^2}{\norm{\mV}}
    \end{align*}
    The second equality is due to the linearity of the trace, and the third is due to the fact that~$\trace \va_l\va_l^\T =2$. 
    From this juncture, we can now bound the operator norm of the matrix-valued variance as follows. First, note that since $\mV$ is a Laplacian matrix, it can be written as~$\mV := \mL  = \mD  - \mM$, where~$\mM \in \R^{n \times n}$ is an adjacency matrix with~$M_{ij} = -c_{ij}$ if~$(i,j) \in \setE$, and zeros along the diagonal, and~$D_{ii} :=  \sum_{l:l\ni i} c_l=d_i(\vc)$. We obtain
    \begin{align*}
        \nu &= \norm{\mV} = \norm{\mL }=\norm{\mD -\mM }\\
        &\overset{(1)}{\leq} \underbrace{\norm{\mD }}_{=\Delta_{\vc} } + \norm{\mM }\\
        &\overset{(2)}{\leq} \Delta_{\vc}  +\sqrt{\norm{\mM }_1\norm{\mM }_{\infty}}\\
        &\overset{(3)}{=} \Delta_{\vc}  + \sqrt{(\max_j\sum_{i}\abs{M_{ij}})(\max_{i} \sum_{j}\abs{M_{ij}})}\\
        &=2\Delta_{\vc},
    \end{align*}
    where step (1) is by the triangle inequality and the fact that $\mD $ is diagonal, and step (2) is because for any matrix $\mA$
    \[ 
    \|\mA\| = \sqrt{\lambda_{\text{max}}\p{\mA^* \mA}} \leq \sqrt{\|\mA^* \mA \|_{\infty}} \leq \sqrt{\|\mA\|_1 \|\mA\|_\infty},
    \] and step (3) is by definition of the matrix norms~$\norm{\cdot}_1,\norm{\cdot}_\infty$. The final equality follows by noting that 
    \[
    \norm{\mM}_1=\norm{\mM}_\infty=\norm{\mD}=\Delta_{\vc}.
    \]

    Furthermore, we can lower bound the spectral norm by considering the Rayleigh quotient; as~$\mV\succeq \vzero$, we can write~$
    \norm{\mV}_2 := \sup_{\norm{\vx}\leq1}\, \vx^\T\mV\vx$. Take~$\vx \gets \ve_i$, then we always have the lower bound
    \[
    \norm{\mV} = \sup_{\norm{\vx}\leq1}\,\vx^\T\mV\vx \geq \sup_{i}\, \ve_i^\T\mV\ve_i = \max_{i}\, d_i(\vc) := \Delta_{\vc}. 
    \]
    Thus,
    \[
    \frac{\sum_{i}d_i(\vc)}{2\Delta(\vc)} \leq \intdim(\mV)  \leq \frac{\sum_{i} d_i(\vc)}{\Delta(\vc)} \leq n-1,
    \]
    since~$\rank(\mV) \leq n-1$ for Laplacian matrices $\mV$.

    We now prepare to invoke the matrix Bernstein inequality~\cite{tropp_concentration_2015}. For non-Hermitian matrices such as~\eqref{eq:rand_Y_conting}, the intrinsic dimension factor~$d$ is given as
    \[
    d =\frac{\trace(\mV)}{\|\mV\|} = \intdim(\mV) \leq \frac{\sum_i d_i(\vc)} {\Delta(\vc)}=\bar{D}.
    \]
    Consequently, for all~$t\geq \sqrt{\nu}+L/3=\sqrt{2\Delta_{\vc}} +2/3$, we have
    \begin{align*}
        \Pr(\|\mYtilde\| \geq t) &\leq 4d\exp\!\p{\frac{-t^2}{2\p{\nu+Lt/3}}}\\
        &\leq 4\p{\frac{\sum_{i}d_i(\vc)}{\Delta(\vc)}}\!\exp\!\p{\frac{-t^2}{4\p{\Delta_{\vc}+t/3}}},
    \end{align*}
    which is the desired result for the tails~\eqref{eq:tail_Y_conting}. To yield the expectation bound~\eqref{eq:expect_Y_conting}, see~\cite[Sec. 7.7.4]{tropp_concentration_2015}. Set~$\bar{D}$ as in~\eqref{eq:norm_tot_deg_crit}.
    Then, a short calculation reveals
    \begin{align*}
    \Expec[]{\|\mYtilde\|} 
    &\leq C\p{\sqrt{\nu\log(1+d)} +L\log(1+d)}\\
    &\leq C\p{\sqrt{2\Delta_{\vc}\log(1+\bar{D})} +2\log(1+\bar{D})}
    \end{align*}
    for some universal constant $C>0$, which is the desired result~\eqref{eq:expec_bound}. This completes the proof of Theorem~\ref{thm:conc_rand_conting}.
\end{proof}

\subsection{Proof of Lemma \ref{lemma:matrix_variance}}
\label{apdx:proof:matrix_variance}
\begin{proof}[Proof of Lemma~\ref{lemma:matrix_variance}]
    By assumption, each $\vg \overset{\sf (d)}{=} \mQ \vz$, $\vb \overset{\sf (d)}{=} \mQ \vz$, where $\vz \overset{\sf(iid)}{\sim} \dNormal(0,\frac{1}{2}\mId)$ is a vector of iid Gaussians and $\mQ_g,\mQ_b \in \R^{n \times n}$ are orthonormal matrices.

First, note that 
\[
\norm{\mM_{ij}} = \norm{\mUpsilon_{ij}} \norm{\mE_{ij}} \leq 4 :=L.
\]

The positive semidefinite upper bound for the matrix-valued variance $\Expec[]{\mF \mF^*} = \Expec[]{\mF^* \mF}$ is then
\begin{align*}
    \Expec[]{\mF \mF^*} &= \sum_{ij \in \setE} \Expec[]{\mM_{ij} \mM_{ij}^*}\\
    &=\sum_{ij \in \setE} \Expec[]{\begin{bmatrix}
     2 \mE_{ij}g_{ij}^2 + 2 \mE_{ij} b_{ij}^2 & \vzero\\
     \vzero & 2 \mE_{ij} b_{ij}^2 + 2 \mE_{ij} g_{ij}^2
    \end{bmatrix}} \\
    &=2 \sum_{ij \in \setE} \begin{bmatrix}
        \mE_{ij}\p{\Expec[]{g_{ij}^2 + b_{ij}^2}} & \vzero\\
        \vzero & \mE_{ij} \p{\Expec[]{g_{ij}^2 + b_{ij}^2}}
    \end{bmatrix}\\
    &\overset{(1)}{\preceq} \frac{2}{n} \sum_{ij \in \setE} \begin{bmatrix}
        \mE_{ij} & \vzero\\
        \vzero & \mE_{ij}
    \end{bmatrix}\\
    &\overset{(2)}{=} \frac{2}{n} \begin{bmatrix}
        \mA^\T \mA & \vzero\\
        \vzero & \mA^\T \mA
    \end{bmatrix}= \frac{2}{n} \mId_{2} \otimes \mA^\T \mA,
\end{align*}
where step (1) stems from $\E{g_{ij}^2} \leq 1$ and $\E{b_{ij}^2} \leq 1$ by assumption of the uniform distribution over the unit sphere.
\end{proof}

\subsection{Proof of Theorem~\ref{thm:bounded-lcpf}}
\label{apdx:proof:bounded-lcpf}
\begin{proof}[Proof of Theorem~\ref{thm:bounded-lcpf}]
The positive semidefinite upper bound for the matrix-valued variance $\Expec[]{\mF \mF^*} = \Expec[]{\mF^* \mF}$ is then
\begin{align*}
    \Expec[]{\mF \mF^*} &= \sum_{ij \in \setE} \Expec[]{\mM_{ij} \mM_{ij}^*}\\
    &=2 \sum_{ij \in \setE} \begin{bmatrix}
        \mE_{ij}\p{\Expec[]{g_{ij}^2 + b_{ij}^2}} & \vzero\\
        \vzero & \mE_{ij} \p{\Expec[]{g_{ij}^2 + b_{ij}^2}}
    \end{bmatrix}\\
    &\overset{(1)}{\preceq} 4 \Delta^2 \sum_{ij \in \setE} \begin{bmatrix}
        \mE_{ij} & \vzero\\
        \vzero & \mE_{ij}
    \end{bmatrix}\\
    &\overset{(2)}{=} 4 \Delta^2 \begin{bmatrix}
        \mA^\T \mA & \vzero\\
        \vzero & \mA^\T \mA
    \end{bmatrix}\\
    &= 4 \Delta^2 \mId_{2} \otimes \mA^\T \mA,
\end{align*}
where step (1) stems from the fact that $\Expec[]{g_{ij}^2} \leq \Delta^2$ and $\Expec[]{b_{ij}^2} \leq \Delta^2$ by assumption of bounded model uncertainty.
Then, the matrix variance statistic is bounded as
\[
\nu \leq  4 \Delta^2 \norm{\mA^\T \mA} \leq 4 \Delta^2 n.
\]
Now, the matrix Bernstein inequality completes the proof.
\end{proof}
\vspace{-1em}

\bibliographystyle{ieeetr}
\bibliography{refs,refs-extra,zotero}
\balance

\endgroup
\end{document}